\documentclass[aps, reprint, amsmath,amssymb, pra, superscriptaddress, longbibliography]{revtex4-2}
\usepackage{svg}
\usepackage{amsmath,amssymb,amsthm,mathtools,mathrsfs}
\usepackage{color}
\usepackage{xr-hyper}
\usepackage{graphicx}
\usepackage{dcolumn}
\usepackage{bm}
\usepackage{physics}
\usepackage{hyperref}
\usepackage{url}
\usepackage{amsmath}
\usepackage{tikz-cd}
\usepackage{notes2bib}
\usepackage{multirow}
\usepackage{mathrsfs}
\usepackage{tabularx}
\usepackage{ dsfont }
\usepackage{svg}
\usepackage{verbatim}
\usepackage{amssymb}
\usepackage{amsfonts}              
\usepackage{amsthm}                
\usepackage{mathtools}
\usepackage{qcircuit}
\usepackage[normalem]{ulem}
\usepackage{xcolor}
\usepackage{csquotes}
 \usepackage [latin1]{inputenc}
\usepackage[many]{tcolorbox}

\usepackage[clock]{ifsym}
\usepackage{fontawesome}
\usepackage{manfnt}

\usepackage{enumitem}

\definecolor{amaranth}{rgb}{0.9, 0.17, 0.31}

\makeatletter
\newcommand*{\addFileDependency}[1]{
  \typeout{(#1)}
  \@addtofilelist{#1}
  \IfFileExists{#1}{}{\typeout{No file #1.}}
}
\makeatother

\newcommand*{\myexternaldocument}[1]{%
    \externaldocument{#1}%
    \addFileDependency{#1.tex}%
    \addFileDependency{#1.aux}%
}

\myexternaldocument{SM}

\DeclarePairedDelimiterX{\barpair}[2]{(}{)}{%
  #1\;\delimsize\|\;#2%
}
\theoremstyle{definition}
\newtheorem{theorem}{Theorem}

\newtheorem{remark}{Remark}

\newtheorem{lemma}{Lemma}

\newtheorem{definition}{Definition}

\newcommand{\id}{\text{id}}

\newcommand{\mcal}[1]{\mathcal{#1}}

\newcommand{\jami}{Jamio{\l}kowski }

\newcommand{\mds}[1]{\mathds{#1}}

\newcommand{\E}{\mathcal{E}}
\newcommand{\F}{\mathcal{F}}
\newcommand{\M}{\mathcal{M}}

\newcommand{\swap}{{\tt SWAP}}

\usepackage[normalem]{ulem}
\bibnotesetup{note-name    = ,use-sort-key = false}
\begin{document}

\title{Multipartite quantum states over time from two fundamental assumptions}

\author{Seok Hyung Lie}
\email{seokhyung@unist.ac.kr}
\affiliation{Department of Physics, Ulsan National Institute of Science and Technology (UNIST), Ulsan 44919, Republic of Korea}

\author{James Fullwood}
\email{fullwood@hainanu.edu.cn}
\affiliation{School of Mathematics and Statistics, Hainan University, Haikou, Hainan Province, 570228, China}

\date{\today}

\begin{abstract}
The theory of quantum states over time extends the density operator formalism into the temporal domain, providing a unified  of treatment of timelike and spacelike separated systems in quantum theory. Although recent results have characterized quantum states over time involving two timelike separated systems, it remains unclear how to consistently extend the notion of quantum states over time to multipartite temporal scenarios, such as those considered in studies of Leggett-Garg inequalities. In this Letter, we show that two simple assumptions uniquely single out the Markovian multipartite extension of bipartite quantum states over time, namely, linearity in the initial state and a quantum analog of conditionability for multipartite probability distributions. As a direct consequence of our result, we establish a canonical correspondence between multipartite QSOTs and Kirkwood-Dirac type quasiprobability distributions, which we show opens up the possibility of experimentally verifying the temporal correlations encoded in QSOTs via the recent experimental technique of simulating quasiprobability known as \emph{quantum snapshotting}. 
\end{abstract}

\pacs{Valid PACS appear here}
\maketitle

\textit{Introduction}---
In general relativity, Lorentzian manifolds provide a single mathematical formalism with which to treat both space and time. In quantum theory however, there exists a disparity in the mathematical formalisms one uses to investigate correlations across space on the one hand, and correlations across time on the other. In particular, while correlations across regions of space are encoded by multipartite density operators, temporal correlations arise out of interactions described by quantum channels, yielding multipartite probability distributions which may not in general be described by the density operator formalism. Is such a disparity between the ways in which quantum theory treats space and time merely an artifact of the historical development of quantum theory, or is it a reflection of some fundamental, intrinsic property of nature which forces quantum theory to be incompatible with the notion of spacetime?

Such a question was formally addressed in Ref.~\cite{HHPBS17}, where the notion of a \emph{quantum state over time} (QSOT) was introduced to give a more balanced treatment of space and time in quantum theory. A QSOT was first proposed as a bipartite hermitian operator of unit trace associated with a two-time process corresponding to a single run of a quantum channel, thus providing a spatiotemporal analog of a bipartite density operator. Contrary to density operators however, QSOTs are not positive in general, as the negative eigenvalues of a QSOT serve as a witness to temporal correlations which have no spatial analog.

In Ref.~\cite{FuPa22}, a construction of QSOTs was established in accordance with desiderata put forth in Ref.~\cite{HHPBS17}, and the operational significance of such QSOTs in terms of two-time correlation functions was recently addressed in Ref.~\cite{FuPa24}. For systems of qubits, bipartite QSOTs have been shown to encompass the construction of pseudo-density matrices~\cite{FJV15,LQDV23,HHPBS17}, which have found applications to a wide variety of topics including quantum channel capacity and the black hole information problem~\cite{Mar2020,MVVAPGDG21,JSK23,song23,PZV2019,ZPTGVF18,FZPV24}. The QSOT formalism has also been utilized in Ref.~\cite{FuPa22a} to formulate a quantum Bayes' rule, leading to a notion of operational reversibility for open quantum systems~\cite{FuPa24a}. Recently the QSOT formalism has also been utilized in the study of temporal entanglement \cite{milekhin2025observable} and the quantum transport problem \cite{hoogsteder2025approach}.

Although bipartite QSOTs have indeed provided a natural extension of the density operator formalism into the temporal domain, there are various issues in regards to unambiguously generalizing QSOTs to a multipartite setting~\cite{LeSp13}, leading some to believe that a well-defined multipartite construction of QSOTs is not possible~\cite{RCGWF18}. In this Letter, we take a cue from the derivation of special relativity, and derive an explicit expression for multi-partite QSOTs based on two simple assumptions. The first assumption is that a QSOT should depend linearly on the initial state of a dynamical process, in accordance not only with the convex structure of quantum state spaces, but also with the statistical behavior of dynamical quantities such as work and heat. The second and final assumption is motivated by the fact that equations of motion should be uniquely determined by initial conditions subject to dynamical laws, which we formulate in terms of a conditionability condition for QSOTs. 

While our assumptions are quite simple, we provide further justification for them by showing how the QSOTs derived from our assumptions yield a spatiotemporal generalization of a quantum Markov chain, a notion which was previously only defined in terms of a (spacelike) density operator \cite{accardi2020quantum, datta2015quantum, fawzi2015quantum, sutter2018approximate}. Moreover,  as a consequence of our results, we establish a canonical correspondence between QSOTs and multipartite Kirkwood-Dirac distributions associated with POVMs measured over time, which we show admits an operational interpretation in terms of a recently established technique for measuring quantum temporal correlations referred to as \emph{quantum snapshotting}~\cite{wang2024snapshotting}. Finally, we conclude by showing how the QSOTs singled out by our assumptions encode all relevant correlations associated with a Leggett-Garg scenario. Interestingly, we find that QSOTs yield a purely quantum description of Leggett-Garg scenarios without invoking any notion of state collapse.

\textit{Preliminaries}---
Given finite-dimensional quantum systems $A_0,\ldots,A_n$, the associated joint system with Hilbert space $\bigotimes_{i=0}^{n}\mathcal{H}_{A_i}$ will be denoted by $A_0\cdots A_n$. We will mathematically identify the system $A_i$ with the algebra of linear operators on $\mathcal{H}_{A_i}$, and a completely positive, trace-preserving linear map  $\E_{i}:A_{i-1}\to A_i$ will be referred to as a \emph{quantum channel}. An $n$-tuple $(\E_1,\ldots,\E_n)$ of quantum channels $\E_i:A_{i-1}\to A_i$ will be referred to as an \textit{$n$-chain} from $A_0$ to $A_n$, and a density operator $\rho$ of the initial system $A_0$ will be referred to as an \emph{initial state} for the $n$-chain $(\E_1,\ldots,\E_n)$. The truncations of an $n$-chain $\bm{\E}=(\E_1,\E_2,\dots,\E_n),$ from below and above will be denoted by
\begin{align*}
    \underline{\bm{\E}}:=(\E_2,\E_3,\dots,\E_n)\quad \text{and} \quad \overline{\bm{\E}}:=(\E_1,\E_2,\dots,\E_{n-1})\,.
\end{align*}

The following definition yields a quantum analog of the multipartite distribution associated with a sequence $X_0,\ldots,X_n$ of classical random variables.

\begin{definition}[Quantum state over time] \label{QSOTDFX67}
A \textit{$\star$-product} (or \textit{spatiotemporal product}) is a binary operation that maps every $n$-chain $\bm{\E}$ from $A_0$ to $A_n$ with initial state $\rho$ to an operator $\bm{\E}\star\rho$ on $\mathcal{H}_{A_0}\otimes \cdots \otimes \mathcal{H}_{A_n}$ such that
\begin{align*}
\Tr_{A_0}[\bm{\E}\star\rho]= \underline{\bm{\E}}\star\E_1(\rho) \quad\text{and}\quad
\Tr_{A_n}[\bm{\E}\star\rho] = \overline{\bm{\E}}\star\rho \, .
\end{align*}
In such a case, $\bm{\E}\star\rho$ is said to be the \emph{quantum state over time} associated with $n$-chain $\bm{\E}$ and the initial state $\rho$. 
\end{definition}
For $n=1$, it follows that $\overline{\bm{\E}}$ and $\underline{\bm{\E}}$ are simply the empty chain, so that in such a case we set $\underline{\bm{\E}}\star\E_1(\rho)=\E_1(\rho)$ and $\overline{\bm{\E}}\star\rho=\rho$ in Definition~\ref{QSOTDFX67}. The restriction of a $\star$-product to $n$-chains for a fixed $n$ will be referred to as a \emph{$\star$-product on $n$-chains}. While a $\star$-product on 1-chains was referred to as a ``QSOT function" in Refs.~\cite{FuPa22a,LiNg23,PFBC23}, we use the $\star$-product terminology to be more consistent with the notation $\bm{\E}\star \rho$. 

\textit{Bipartite QSOTs}---
 In their formulation of quantum theory as a causally neutral theory of Bayesian inference~\cite{LeSp13}, Leifer and Spekkens introduced the spatiotemporal product on $1$-chains $\star_{\text{LS}}$ given by 
 \[
 \E\star_{\text{LS}} \rho=(\sqrt{\rho}_A\otimes \mathds{1}_{B})\mathscr{J}[\E](\sqrt{\rho}_A\otimes \mathds{1}_{B})\, ,
 \]
 where $\mathscr{J}[\E]=\sum_{i,j}|i\rangle \langle j|_A \otimes \E(|j\rangle \langle i|)_B$ is the \emph{Jamio{\l}kowski matrix} of a quantum channel $\E:A\to B$ (which is not to be confused with the Choi matrix, which differs from the \jami matrix by a partial transpose). While the LS-product $\star_{\text{LS}}$ has appeared extensively in the literature on quantum Bayesian inference \cite{bai2024bayesian, bai2024quantum, aw2024role, AwBuSc21}, there are significant drawbacks with the construction. In particular, the LS-product is not linear in $\rho$, and perhaps most importantly, it does not extend in any natural way to multipartite temporal scenarios.

While there are indeed various state-linear $\star$-products on 1-chains which admit natural extensions to multipartite scenarios, such as the left and right products $\star_{\text{L}}$ and $\star_{\text{R}}$ given by $\E \star_{\text{L}} \rho = (\rho \otimes \mathds{1}_B)\mathscr{J}[\E]$ and $\E \star_{\text{R}} \rho = \mathscr{J}[\E] (\rho \otimes \mathds{1}_B)$, such constructions are based on mathematical simplicity, and are not derived from physical principles. It had then been an open problem to derive a $\star$-product from simple physical assumptions~\cite{HHPBS17}, and such a derivation was achieved in Refs.~\cite{LiNg23,PFBC23}. Moreover, while the sets of axioms employed in Refs.~\cite{LiNg23,PFBC23} were distinct, the resulting $\star$-product was the same. 

The $\star$-product on 1-chains uniquely derived from physical principles was referred to as the \textit{Fullwood-Parzygnat} (FP)-product in Ref. ~\cite{LiNg23}, which is given by
\begin{equation} \label{eqn:FPfunc}
    \E\star_{\text{FP}}\rho= \frac{1}{2}\Big\{\rho_A\otimes\mds{1}_B,\mathscr{J}[\E]\Big\}
\end{equation}
for all states $\rho$ on $A$ and quantum channels $\E:A\to B$, where $\{\cdot\, ,\cdot\}$ denotes the anti-commutator. While such bipartite constructions of QSOTs have proven useful for treating space and time on an equal footing, it turns out that there are a multitude of ways in which one may extend a $\star$-product on $1$-chains to $n$-chains~\cite{Fu23,Fu23a}. In fact, in the case of the FP-product, we show in Appendix~\ref{app:B} that the number of $\star$-products on $n$-chains which reduce to $\star_{\text{FP}}$ on 1-chains grows exponentially with $n$. This is highly undesirable: Imagine how confusing it could have been if there were $2^n$ different constructions of the density operator for $n$-partite systems. As such, it would be desirable if we could find simple assumptions which single out a unique extension of quantum states over time to the multipartite setting, without imposing one by fiat. 

\textit{Multipartite extension of QSOTs}---
The main result of this Letter is that two fundamental assumptions imply that $\star$-product is completely determined by its restriction to $1$-chains, thus yielding a unique multi-partite extension of bipartite QSOTs. Our first assumption is simply the requirement that a $\star$-product should be consistent with statistical mixtures: if there is uncertainty about the initial state preparation, this uncertainty should carry over to the overall dynamics to the same degree. Mathematically, this means that $\bm{\E}\star\rho$ is convex-linear in $\rho$, i.e., for every $n$-chain $\bm{\E}$, we assume $\bm{\E}\star (\sum_ip_i\rho^i)=\sum_ip_i (\bm{\E}\star \rho^i)$. This convex-linearity in density operator can be extended to linearity in general operator with the standard procedure \cite{wilde2013quantum, SM}, and thus will be referred to as \textit{state-linearity}. The second assumption is motivated from the fact that it is customary to describe the evolution of a physical system in terms of initial conditions subject to dynamical laws. In the context of classical stochastic dynamics associated with a with a sequence of random variables $X_0,\dots,X_n$, such a principle manifests itself at the level of the associated joint distribution $P(x_0,\ldots,x_n)$ via the equation
\begin{equation}\label{eqn:condition}
    P(x_0,\ldots,x_n)=P(x_0) \cdot P(x_1,\ldots,x_n|x_0)\, ,
\end{equation}
where $P(x_0)$ is viewed as the initial state of the process. We refer to this property of multipartite probability distributions as \emph{conditionability}, which we now generalize to the setting of QSOTs.

\begin{definition}[Conditionability] A $\star$-product is said to be \textit{conditionable} if and only if for every state $\rho$ on $A_0$, there exists a linear map $\Theta_\rho:A_0\to A_0$ such that for every $n$-chain $\bm{\E}$ from $A_0$ to $A_n$,
\begin{equation} \label{eqn:quantcond}
\bm{\E} \star \rho = (\Theta_\rho \otimes \id_{A_1\cdots A_n}) (\bm{\E} \star \mds{1}_{A_0})\, .
\end{equation}
\end{definition}

The appearance of the map $\Theta_{\rho}$ on the RHS of equation \eqref{eqn:quantcond} is an operator analog of being able to factor out the contribution of a prior $P(x_0)$ from a multipartite distribution as in \eqref{eqn:condition}, while the operator $\bm{\E} \star \mds{1}_{A_0}$ is a quantum analog of the conditional distribution $P(x_1,\ldots,x_n|x_0)$. 

Before stating our main result we set some notation in place. For a state-linear $\star$-product and an arbitrary quantum channel $\E$, $\E\,\star$ can be understood as a linear map $\rho \mapsto \E \,\star\, \rho$. Hence, given an $n$-chain $(\E_1,\ldots,\E_n)$, after omitting the tensor product with the identity channels, $\E_k\star ((\E_1\ldots,\E_{k-1})\star \rho)$ is well-defined for any $2\leq k \leq n$.

\begin{theorem}[Unique multi-partite extension of QSOTs] \label{thm:iterfrcond}
A $\star$-product which is state-linear and conditionable is uniquely determined by its restriction to 1-chains. In particular, for every $n$-chain $\bm{\E}=(\E_1,\ldots,\E_n)$ with initial state $\rho$,
\begin{equation} \label{eqn:iterativity}
\bm{\E}\star \rho=\E_n \star ( \E_{n-1} \star ( \cdots \star(\E_1 \star \rho)))\, .
\end{equation}
\end{theorem}

The proof of Theorem~\ref{thm:iterfrcond} is given in the Supplementary Material \cite{SM}. Due to its structure, the canonical extension of the QSOT characterized in Theorem \ref{thm:iterfrcond} will be referred to as the \textit{Markovian extension}, the meaning of which will be elaborated upon later. (See FIG. \ref{fig:enter-label} (a)).
In what follows, we explore some further implications of Theorem~\ref{thm:iterfrcond}. 

\textit{Multipartite KD distributions via QSOTs}---
In quantum theory, incompatible observables cannot be measured simultaneously, as evidenced by the non-existence of well-defined joint probability distributions for their measurement outcomes. In 1933 and 1945 respectively, Kirkwood \cite{kirkwood1933quantum} and Dirac \cite{dirac1945analogy} independently proposed a way to circumvent this issue by allowing for complex-valued probabilities, which are now commonly referred to as \emph{quasi}probability distributions. This construction, known today as the \textit{Kirkwood-Dirac} (KD) distribution, plays a fundamental role in the study of non-classical features of quantum temporal correlations, especially in the context of quantum thermodynamics \cite{Lostaglio2023kirkwooddirac, arvidsson2024properties}. In this section, we establish a precise correspondence between $\star$-products and classes of quasiprobability distributions, and then use Theorem \ref{thm:iterfrcond} to uniquely extend the latter to the multipartite setting.

Let $\{M_i\}\subset A$ and $\{N_j\}\subset B$ be POVMs which are to be measured at times $t_1$ and $t_2>t_1$. It then follows that for every $\star$-product on 1-chains the elements $Q_{AB}(i,j)\in \mathbb{C}$ given by
\begin{equation} \label{QPXDSX67}
Q_{AB}(i,j)=\Tr\Big[\E\star \rho\, (M_i\otimes N_j)\Big]
\end{equation}
form a quasiprobability distribution for every channel $\E:A\to B$ with initial state $\rho$ on $A$. The marginals $Q_A(i)=\sum_j Q_{AB}(i,j)$ and $Q_B(j)=\sum_i Q_{AB}(i,j)$ are probability distributions given by $Q_A(i)=\Tr[\rho \hspace{0.2mm} M_i] \quad \text{and} \quad Q_B(j)=\Tr[\E(\rho) N_j]\,$, which are the standard probabilities associated with the Born rule. It turns out that Eq.~\eqref{QPXDSX67} can be used to establish an explicit correspondence between bipartite QSOTs and bipartite quasiprobability distributions, and this correspondence can be easily generalized to the $n$-step setting (see Appendix \ref{app:A} for more detail). We note that under such a correspondence, the state-linear $\star$-products $\star_{\text{L}}$ and $\star_{\text{FP}}$ correspond to the Kirkwood-Dirac and Margenau-Hill distributions, respectively.

 If a $\star$-product is state-linear, then it follows from the no-go theorem for quasiprobability distributions appearing in Ref.~\cite{Lostaglio2023kirkwooddirac} that the non-positivity of generic $Q_{AB}(i,j)$ as given by \eqref{QPXDSX67} is inevitable. This is a reflection of the fact that the QSOT $\E\star \rho$ will generally admit negative eigenvalues, thus the non-positivity of $\E\star \rho$ is an indicator of non-classicality for the temporal correlations established by $\E$ and $\rho$. 

An immediate implication of Theorem~\ref{thm:iterfrcond} is that quasiprobability distributions as given by \eqref{QPXDSX67} canonically extend to an arbitrary number of POVMs measured over time, provided that the associated $\star$-product is conditionable and state-linear \footnote{In particular, we note that Theorem~\ref{thm:iterfrcond} does not provide a unique extension of the Leifer-Spekkens distribution since $\star_{\text{LS}}$ is not state-linear}. We also note that while previous works have addressed the issue of extending quasiprobability distributions to multipartite scenarios~\cite{yunger2018quasiprobability, arvidsson2020quantum}, they were primarily formal and not derived from basic assumptions. 

The fact that QSOTs encode quasiprobability distributions implies that the many uses of the Kirkwood-Dirac and Margenau-Hill distributions in quantum chaos, quantum thermodynamics, and quantum metrology \cite{arvidsson2024properties} carry over directly to the QSOT. Furthermore, Theorem~\ref{thm:iterfrcond} provides a firm mathematical foundation for studies utilizing certain multipartite KD distributions. For example, in post-selected quantum metrology \cite{arvidsson2020quantum}, advantages beyond the Heisenberg limit arise only when the tripartite KD distribution associated with the scheme exhibits negativity. Finding necessary conditions for such advantages remains an open problem, and we expect that operator properties of QSOTs (such as their eigenvalues) will provide novel insights in this direction.

\begin{figure}
    \centering
    \includegraphics[width=\linewidth]{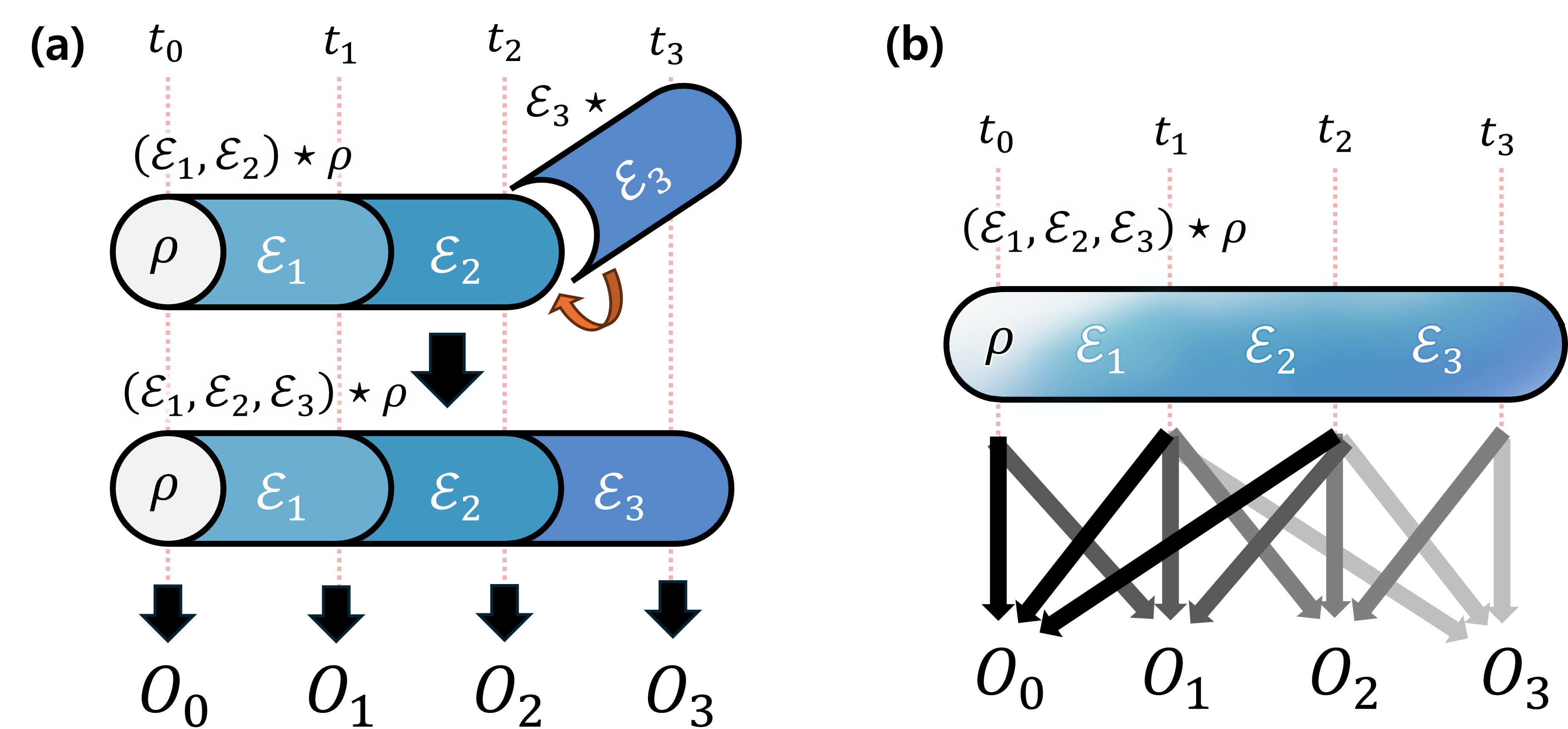}
    \caption{(a) The Markovian extension of a QSOT uniquely characterized in Theorem \ref{thm:iterfrcond} allows for sampling observables $O_i$ at time $t_i$ through temporally localized interventions, i.e., quantum snapshotting. (b) Contrarily, non-Markovian QSOTs do not have such a simple decomposition and sampling an observable at each time requires global access to multiple time-steps.}
    \label{fig:enter-label}
\end{figure}

\textit{Markovianity and quantum snapshotting}---
In classical probability, a sequence of random variables $X_0,\ldots,X_n$ is said to form a \emph{Markov chain} if the associated conditional distribution $P(x_i|x_{i-1},\ldots,x_0)$ satisfies
\begin{equation}\label{MKVCXDN87}
P(x_i|x_{i-1},\ldots,x_0)=P(x_i|x_{i-1}) \qquad  1<i\leq n\, .
\end{equation}
For a quantum system evolving under an $n$-chain $(\E_1,\ldots,\E_n)$, the associated quasiprobabilities as established in the previous section do not satisfy the Markovian condition \eqref{MKVCXDN87} (see Appendix \ref{app:NonMarkov} for an explicit example). Therefore, at the level of quasiprobabilities the Markovianity condition in the classical sense fails, reflecting the fact that a measurement at time $t_i$ may have causal influence on measurements at times $t>t_{i+1}$. 

On the other hand, the formula
\[
(\E_1, \E_2)\star \rho = \E_2 \star (\E_1 \star \rho)
\]
as given by Eq.~\eqref{eqn:iterativity} of Theorem~\ref{thm:iterfrcond} in the case of 2-chains is more reminiscent of the classical expression $P(x_0,x_1,x_2)=P(x_2|x_1)P(x_1|x_0)P(x_0)$ associated with a Markov chain $X_0\to X_1\to X_2$. Moreover, we can see that the 3-time QSOT $(\E_1, \E_2) \star \rho$ is obtained from the 2-time QSOT $\E_1 \star \rho$ by the mapping $\E_2\,\star$ acting on $A_1$ alone, which is consistent with the definition of a quantum Markov chain of Refs. \cite{accardi2020quantum, datta2015quantum, fawzi2015quantum, sutter2018approximate} for spatial correlations. Such an observation suggests that in the temporal domain, quantum Markovianity is a property naturally formulated at the level of QSOTs rather than at the level of quasiprobabilities, which justifies referring to Eq.~\eqref{eqn:iterativity} of Theorem~\ref{thm:iterfrcond} as the Markovian extension of bipartite QSOTs.
As such, an immediate corollary of Theorem \ref{thm:iterfrcond} is that one can formulate quantum Markov chains directly at the level of the quantum state, covering both spatial and temporal cases, contrary to higher-order constructions such as process matrices \cite{OrCe15, OrCe16}, process tensors \cite{pollock2018operational, pollock2018non}, or quantum combs \cite{chiribella2008quantum}.

The Markovian nature of multipartite QSOTs is also crucial for the possibility of measuring the expectation values encoded by a QSOT by \emph{quantum snapshotting}, which is a recently proposed technique for experimental verification of quasiprobabilities \cite{wang2024snapshotting}. In particular, consider the expectation value $\langle \mathscr{O}_1,\ldots,\mathscr{O}_n \rangle$ of observables $\mathscr{O}_i$ associated with the $i$-th step in an $n$-chain QSOT $(\E_1,\E_2,\dots,\E_n)\star\rho$, calculated as $\Tr[(\mathscr{O}_0\otimes\cdots\otimes \mathscr{O}_n)(\E_1,\E_2,\dots,\E_n)\star\rho]$. As the associated QSOT admits the iterative decomposition as in \eqref{eqn:iterativity}, we can express $\langle \mathscr{O}_0,\dots,\mathscr{O}_n\rangle$ as
\begin{equation} \label{eqn:snapshot}
    \Tr[\mathscr{O}_n (\E_n \circ \M_{\mathscr{O}_{n-1}} \circ \cdots \circ \E_1 \circ \M_{\mathscr{O}_0})(\rho)]\,,
\end{equation}
where $\M_\mathscr{O}=\sum_i \gamma_i(\mathscr{O}) \mcal{I}_i$ is a linear map representing the sampling process of the observable $\mathscr{O}$, which can be expressed as a linear combination of physically implementable quantum instruments $\qty{\mcal{I}_i}$ such that $\sum_i\mcal{I}_i$ is a quantum channel, with complex coefficients $\gamma_i(\mathscr{O})$ processed classically~\cite{wang2024snapshotting} (see Appendix \ref{app:Qsnap} for further details).

Notice that, to obtain the expectation value \eqref{eqn:snapshot}, one can apply the channels $\E_i$ and (the physical implementation of) the snapshotting processes $\mcal{M}_{\mathscr{O}_i}$ in an alternating manner, due to the fact that the action of $\mcal{M}_{\mathscr{O}_i}$ is localized at the $i$-th time-step in \eqref{eqn:snapshot}. This indeed demonstrates the significance of the Markovian extension of QSOTs, as it allows for `snapshotting' in the literal sense: One does not need to retain the memory of all other time-steps to extract information of an observable at time-step $i$. By faithfully capturing Markovianity, such QSOT expansions thus enable feasible experimental verification via quantum snapshotting. (See FIG. \ref{fig:enter-label}.)

\textit{A Leggett-Garg scenario}---Let $\star$ be the Markovian tripartite extension of $\star_{\text{FP}}$, and consider an experiment in which a qubit in state $\rho$ is to be sent through 3 devices $Q_1$, $Q_2$ and $Q_3$ in sequence, where $Q_i$ measures the spin observable $\mathscr{O}_i=\vec{u}_i\cdot \sigma$ with $\vec{u}_1=(1,0,0)$, $\vec{u}_2=(1/2,\sqrt{3}/2,0)$ and $\vec{u}_3=(-1/2,\sqrt{3}/2,0)$. Denote the projectors onto the $\pm 1$-eigenspaces of $\mathscr{O}_i$ by $\Pi_i^{\pm}$, so that $\mathscr{O}_i=\Pi_i^{+}-\Pi_i^{-}$ and $\Pi_i^{+}+\Pi_i^{-}=\mathds{1}$. We assume that for each run of the experiment, exactly one of the devices $Q_k$ is deactivated (alternating between $k=1,2,3$), so that each run of the experiment consists of two measurements $\mathscr{O}_i$ followed by $\mathscr{O}_j$ with $i<j$. We further assume that the qubit evolves trivially as it travels from device $Q_i$ to $Q_{i+1}$, so that the QSOT associated with such a three-time scenario is $\textbf{id}^2\star \rho$, where $\textbf{id}^2:=(\id,\id)$ is the identity 2-chain. Now let $C_{ij}$ be the two-time expectation value associated with runs of the experiment where a measurement of $\mathscr{O}_i$ is followed by a measurement of $\mathscr{O}_j$, so that 
\[
C_{ij}=P_{ij}(1,1)-P_{ij}(1,-1)-P_{ij}(-1,1)+P_{ij}(1,1)\, ,
\] 
where $P_{ij}(\pm 1\,,\pm 1)=\Tr[\Pi_i^{\pm}\rho \,\Pi_i^{\pm}\Pi_j^{\pm}]$. If the values of the observables $\mathscr{O}_1$, $\mathscr{O}_2$ and $\mathscr{O}_3$ are each well-defined for each run of the experiment, regardless of which device $Q_k$ is deactivated, then such a measurement scenario is consistent with the principle of \emph{macrorealism}, which essentially means that the value of an observable is well-defined irrespective of whether or not it was actually measured. If such a measurement scenario is indeed consistent with the principle of macrorealism, and moreover, if it is possible to measure the values of $\mathscr{O}_i$ in a non-invasive manner, then it is straightforward to show 
\begin{equation} \label{LGEQX17}
C_{12}+C_{23}-C_{13}\leq 1\, ,
\end{equation}
which is often referred to as a \emph{Leggett-Garg inequality}~\cite{LeGa85}. However, in such a case we have $C_{12}=C_{23}=-C_{13}=1/2$ (a result which is actually independent of the state $\rho$~\cite{Emary_2013,FV25}), so that $C_{12}+C_{23}-C_{13}=3/2>1$,
thus violating the Leggett-Garg inequality \eqref{LGEQX17}. It then follows that either the assumption of macrorealism or non-invasive measurability must be invalid (or both), implying that there does not exist a well-defined tripartite probability distribution $P(x_1,x_2,x_3)$ associated with the measurement outcomes of each run of the experiment whose marginals are the distributions $P_{ij}(x_i,x_j)$. On the other hand, by direct calculation we find 
\begin{equation} \label{EXPTXTN67}
C_{ij}=\Tr\Big[\textbf{id}^2\star \rho\, (\mathscr{O}_i\otimes \mathscr{O}_j\otimes \mathds{1})\Big]\, ,
\end{equation}
thus, the tripartite QSOT $\textbf{id}^2\star \rho$ encodes the correlations $C_{ij}$ for all $i$ and $j$, despite the absence of a well-defined tripartite probability distribution associated with each run of the experiment. 

Therefore, while such a Leggett-Garg scenario is inconsistent with macrorealism and non-invasive measurability, the QSOT $\textbf{id}^2\star \rho$ nevertheless is a quantum Markov chain encoding all of the relevant correlations of the experiment, independent of which measurements were actually performed. This demonstrates that such a violation of the Leggett-Garg inequality cannot be attributed to a lack of (quantum) Markovianity. Interestingly, the quasiprobabilities $Q_{ij}(x_i,x_j)$ associated with the QSOT $\textbf{id}^2\star \rho$ do not coincide with the probabilities $P_{ij}(x_i,x_j)$ coming from the state-update rule, yet the expectation values corresponding to the weighted sums of $P_{ij}(x_i,x_j)$ and $Q_{ij}(x_i,x_j)$ turn out to be the same. As such, taking \eqref{EXPTXTN67} as a theoretical definition for temporal expectation values yields a spatiotemporal description of Leggett-Garg scenarios without any reference to state collapse. 

\textit{Discussion}---
In this Letter, we have shown that multipartite QSOTs may be uniquely derived from the assumptions of linear dependence on the initial state and quantum conditionability. Such a result not only lays the foundation for a spatiotemporal extension of the density operator formalism, but also for formulating a relativistic theory of quantum information where space and time are treated on equal footing. In particular, the multipartite QSOTs singled out by our assumptions yield a dynamical notion of quantum Markov chain which naturally extends the definition for spacelike separated systems. Moreover, we have also shown how QSOTs encode quasiprobability distributions associated with measurements of POVMs over time. As quasiprobability distributions have emerged as a key tool for characterizing non-classical features of quantum temporal correlations---particularly in fields such as quantum thermodynamics, quantum chaos and metrology~\cite{Lostaglio2023kirkwooddirac}---we expect that the QSOTs derived in this Letter will also come to play a significant role in these developments. Finally, it would be interesting to see if the assumptions employed in this Letter may also be used to derive QSOTs associated with more exotic temporal scenarios, such as non-Markovian dynamics and indefinite causal structures~\cite{Shrikant_2023,OCB12}.

\acknowledgments{SHL is supported by the 2025 Research Fund (1.250007.01) of Ulsan National Institute of Science \& Technology (UNIST) and Institute of Information \& Communications Technology Planning \& Evaluation (IITP) Grant (RS-2023-00227854, RS-2025-02283189, ``Quantum Metrology Theory Based on Temporal Correlation Operators"). JF is supported by the start-up grant of Hainan University for the project ``Spacetime from quantum information".}

\bibliography{main}

\begin{appendix}
\onecolumngrid
\section*{End Matter}
\twocolumngrid
\section{Holistic spatiotemporal products} \label{app:B}
Given $\bm{C}=(C_1,C_2,\dots,C_n)\in \qty{L,R}^n$, we let $\star_{\bm{C}}$ be the spatiotemporal product given by
\begin{equation}
    \bm{\E}\star_{\bm{C}}\rho = \E_n\star_{C_{n}} ( \E_{n-1} \star_{C_{n-1}} (\cdots (\E_1 \star_{C_1} \rho)))\; ,
\end{equation}
so that the bit string $\bm{C}=(C_1,C_2,\dots,C_n)\in \qty{L,R}^n$ determines which spatiotemporal product out of $\star_L$ and $\star_R$ is used at each time-step. Next, taking $\bm{\tilde{C}}=(\tilde{C_1},\tilde{C_2},\dots,\tilde{C_n})$ to be the complementary bit-string to $\bm{C}$, so that $\tilde{C_i}$ is $L$ when $C_i$ is $R$ and vice versa for each $i=1,\dots,n$, we then define a $\star$-product given by 
\begin{equation}
    \bm{\E}\star \rho = \frac{1}{2} \qty(\bm{\E}\star_{\bm{C}}\rho+\bm{\E}\star_{\bm{\tilde{C}}}\rho)\, .
\end{equation}
One can easily check that $\star$ is linear in $\rho$ and reduces to the FP-product $\star_{\text{FP}}$ on 1-chains. Since there are $2^{n-1}$ different such constructions for every $n\geq 1$, the number of $\star$-products on $n$-chains that reduce to $\star_{\text{FP}}$ on 1-chains increases exponentially with $n$. Such an observation highlights the significance of Theorem~\ref{thm:iterfrcond}, which singles out a unique multipartite extension of a $\star$-product on $1$-chains from such a large number of possibilities.

    \section{Equivalence between QSOT and quasiprobability distributions} \label{app:A}
In this section, we establish the one-to-one correspondence between multipartite QSOTs and families of quasiprobability distributions. For the ease of calculation, we will primarily use the expectation value calculated with a quasiprobability distribution, which is also known as the two-time correlation function (or the $n$-time generalization) in the field of open quantum systems \cite{breuer2002theory}, instead of the distribution itself. We start by giving a minimal definition of quasiprobability expectation $\langle \mathscr{O}_0,\mathscr{O}_1,\dots,\mathscr{O}_n\rangle_{(\bm{\E},\rho)}$ of $n+1$ observables $\mathscr{O}_0,\mathscr{O}_1,\dots,\mathscr{O}_n$ for a given $n$-chain $\bm{\E}=(\E_1,\dots,\E_n)$ and an input state $\rho$ of $A_0$. We require two properties: (i) it is linear in each observable $\mathscr{O}_i$ and (ii) it is normalized in the sense that $\langle \mds{1}_{A_0}, \mathscr{O}_1, \dots \mathscr{O}_n\rangle_{(\bm{\E},\rho)} = \langle \mathscr{O}_1,\dots, \mathscr{O}_n\rangle_{(\underline{\bm{\E}},\E_1(\rho))}$, $\langle \mathscr{O}_0,\dots,\mathscr{O}_{n-1},\mds{1}_{A_n}\rangle_{(\bm{\E},\rho)}=\langle \mathscr{O}_0,\dots, \mathscr{O}_{n-1}\rangle_{(\overline{\bm{\E}},\rho)}$ for every $n\geq1$ and  $\langle \mathscr{O}\rangle_{(\emptyset,\rho)}=\Tr[\mathscr{O} \rho]$ for every $\mathscr{O}$ and $\rho$. One recovers the quasiprobability distribution by substituting each observable $\mathscr{O}_i$ with a projector of a projective measurement (or a POVM element). In other words, for any POVM $\{M_x^{(i)}\}$ on $A_i$, the quasiprobability distribution $p_{x_0,\dots,x_n}$ is defined as $p_{x_0,\dots,x_n}:=\langle M_{x_0}, \dots, M_{x_n}\rangle_{(\bm{\E},\rho)}$ and we can see that it satisfies $\sum_{x_0}p_{x_0,\dots,x_n}=p_{x_1,\dots,x_n}$ and $\sum_{x_n} p_{x_0,\dots,x_n}=p_{x_0,\dots,x_{n-1}}$.

It turns out that for every quasiprobability expectation value, there exists a unique operator $\bm{\E}\ast\rho$ on $A_0\dots A_n$ that allows for the following expression for every $\mathscr{O}_i$. One can derive it from the universal property of the tensor product and the Riesz representation theorem, multilinearity of $\langle \mathscr{O}_0,\dots,\mathscr{O}_n\rangle$.
\begin{equation} \label{eqn:O0OnErho}
    \langle \mathscr{O}_0,\dots,\mathscr{O}_n\rangle_{(\bm{\E},\rho)}=\Tr[(\mathscr{O}_0\otimes\dots\otimes \mathscr{O}_n)(\bm{\E}\ast\rho)].
\end{equation}
One can check that $\Tr_{A_0}\bm{\E}\ast\rho = \underline{\bm{\E}}\ast\E_1(\rho)$ and $\Tr_{A_n}\bm{\E}\ast\rho=\overline{\bm{\E}}\ast\rho$. Therefore, $\ast$ qualifies as a spatiotemporal product. Conversely, \eqref{eqn:O0OnErho} is satisfied by an arbitrary spatiotemporal product $\star$. Hence, we can conclude that there exists a 1-1 correspondence between quasiprobability distributions and spatiotemporal products because the QSOT is the operator representation of a given quasiprobability distribution.

\section{Non-Markovianity of quasiprobability distribution} \label{app:NonMarkov}
For example, let $A,B$ and $C$ denote qubit systems at times $t_A<t_B<t_C$, and consider $\textbf{id}^2\star \pi_A:=(\id_{B|A},\id_{C|B})\star \pi_{A}$, where $\pi_{A}$ is maximally mixed and $\star$ denotes the Markovian extension of the FP-product to 2-chains. The tripartite quasiprobability distribution for measurements with respect to orthonormal bases $\qty{\ket{a_i}}$, $\qty{\ket{b_j}}$ and $\qty{\ket{c_k}}$ of systems $A$, $B$ and $C$ associated with $\textbf{id}^2\star \pi_A$ is the Markovian extension of the bipartite MH distribution, and the associated conditional distribution $Q_{C|AB}(k|i,j):=Q_{ABC}(i,j,k)/Q_{AB}(i,j)$ is given by
\begin{equation}
    Q_{C|AB}(k|i,j)=\Re\qty[\frac{\braket{c_k}{b_j}\braket{a_i}{c_k}}{\braket{a_i}{b_j}}]\, .
\end{equation}
However, a direct calculation yields $Q_{C|B}(k|j)=|\braket{b_j}{c_k}|^2\neq Q_{C|AB}(k|i,j)$ in general.

\section{Derivation of Eq. \eqref{eqn:snapshot}} \label{app:Qsnap}

In this section, we derive the expression \eqref{eqn:snapshot} for the expectation value $\langle \mathscr{O}_0,\ldots,\mathscr{O}_n \rangle$ of observables $\mathscr{O}_0,\dots,\mathscr{O}_n$ at times $t_0,\dots,t_n$ for the class of multipartite QSOTs characterized in Theorem \ref{thm:iterfrcond}. First, we observe that a conditionable $\star$-product is also decomposable which means that for arbitrary quantum channel $\E:A \to B$, $\E\star \rho$ decomposes into the form of $\E\star\rho=(\Theta_\rho\otimes \E)(\swap)$ where $\swap$ is the swap operator on two copies of system $A$, $A$ and $A'$, with a function $\Theta_\rho$ linear in $\rho$ (which can be different from the state rendering function in \eqref{eqn:quantcond} up to an invertible linear map. See Lemma 4 of Supplementary Material \cite{SM}.) Then it follows that $\Tr_A [\mathscr{O} (\E\star \rho)]=\Tr[(\mathscr{O}\otimes\mds{1})(\id_A\otimes \E\circ \Theta^\ddag_\rho)(\swap)]=\E\circ \Theta^\ddag_\rho (\mathscr{O})$ where $\Theta^\ddag_\rho$ is the unique linear map such that $(\Theta_\rho\otimes\id_{A'})(\swap)=(\id_A\otimes \Theta^\ddag_\rho)(\swap)$. (See Lemma 3 of Supplementary Material \cite{SM}.) By state-linearity of the $\star$-product, the mapping $\rho\mapsto \Theta^\ddag_\rho(\mathscr{O})$ is a linear map, so we define a linear map $\mcal{M}_{\mathscr{O}}(\rho):=\Theta^\ddag_\rho(\mathscr{O}).$

Next, note that the mapping $\varrho \mapsto \Tr[(\mathscr{O}_0\otimes \dots \otimes \mathscr{O}_n)\varrho]$ decomposes into the tensor product of functionals $\mcal{F}_0\otimes \mcal{F}_1\otimes \cdots \otimes \mcal{F}_n$, where $\mcal{F}_i(\sigma):=\Tr_{A_i}[\mathscr{O}_i\sigma]$ is a functional define on $A_i$, because the total trace decomposes into the partial traces as $\Tr=\Tr_{A_0}\otimes\cdots\otimes\Tr_{A_n}$. Now, for a given $n$-chain $\bm{\E}=(\E_1,\E_2,\dots,\E_n)$, if $\E\star \rho$ decomposes into $(\E_n^\star\circ\cdots\circ\E_1^\star)(\rho)$ where $\E^\star_i$ is the linear map $\varrho \mapsto \E\star\varrho$, then it follows that $\Tr[(\mathscr{O}_0\otimes \dots \otimes \mathscr{O}_n)(\bm{\E}\star\rho)]=(\mcal{F}_0\otimes\cdots\otimes\mcal{F}_n)(\E^\star_n\circ\cdots\circ\E^\star_1(\rho))$. Observe that this equals to
\[
\mcal{F}_n\circ(\mcal{F}_{n-1}\circ\E_n^\star)\circ\cdots\circ(\mcal{F}_0\circ\E_1^\star)(\rho)\, ,
\]
since $\mcal{F}_i$ commutes with $\E^\star_j$ with $j>i+1$ as the former only acts on $A_i$, while the latter acts on $A_k$s with $k>i$. Due to the argument given in the previous paragraph, we have $\mcal{F}_i\circ \E^\star_{i+1}=\E_{i+1}\circ \mcal{M}_{\mathscr{O}_i}$ for $i=0,\dots,n-1$. Hence, we can conclude that the expectation value $\langle \mathscr{O}_1,\ldots,\mathscr{O}_n \rangle$ coincides with 
\[
\Tr[\mathscr{O}_n (\E_n \circ \M_{\mathscr{O}_{n-1}} \circ \cdots \circ \E_1 \circ \M_{\mathscr{O}_0})(\rho)]\, ,
\]
as claimed in the main text~\eqref{eqn:snapshot}.

\end{appendix}

\clearpage
\onecolumngrid
\begin{center}
    {\large \textbf{Supplemental Material for \\
    ``Multipartite quantum states over time from two fundamental assumptions"}}
\end{center} 
\vspace{1em}

We now prove our main theorem, namely, that a spatiotemporal product which is state-linear and conditionable is necessarily given by the iterative product. First, we set some notation and conventions. 

Given a finite-dimensional quantum system $A$, for ease of notation we will denote the algebra of linear operators on the associated Hilbert space $\mathcal{H}_A$ also by $A$. The set of states on $A$ consisting of all unit trace positive operators on $\mathcal{H}_A$ will be denoted by $\mathfrak{S}(A)$, and we note that every $a\in A$ may be written as a linear combination $a=\sum_i\mu_i \rho^i$ with $\mu_i\in \mathbb{C}$ and $\rho^i\in \mathfrak{S}(A)$, which will be referred to as a \emph{state decomposition} of $a$. We note that a linear map $\E:\mathfrak{S}(A)\to B$ uniquely extends to a linear map $\widetilde{\E}:A\to B$. Moreover, given finite-dimensional quantum systems $A_0,\ldots,A_n$, the tensor product $A_0\otimes \cdots \otimes A_n$ of the associated algebras will be denoted by $A_0\cdots A_n$. The set of linear maps from $A$ to $B$ will be denoted by $\textbf{Lin}(A,B)$. We recall that the \emph{\jami isomorphism} is the linear map $\mathscr{J}:\textbf{Lin}(A,B)\to AB$ given by
\[
\mathscr{J}[\E]=(\id_A\otimes \E)(\swap)\, ,
\] 
where $\swap=\sum_{i,j}E_{ij}\otimes E_{ji}$ is the swap operator, $E_{ij}=|i\rangle \langle j|$ and $\{|i\rangle\}$ is an orthonormal basis of $\mathcal{H}_A$ (the formula for the swap operator  is independent of a choice of basis). It then follows that for every bipartite operator $M_{AB}\in AB$, the linear map $\mathscr{J}^{-1}[M_{AB}]:A\to B$ is given by
\begin{equation}
    \mathscr{J}^{-1}[M_{AB}](\rho)=\Tr_A[(\rho\otimes\mds{1})M_{AB}]\, .
\end{equation}
We recall that a linear map $\mathfrak{B}:A_0\to A_0A_1$ (with $A_0=A_1$) is said to be a \emph{broadcasting map} if and only if for all $\rho\in \mathfrak{S}(A_0)$,
\[
\Tr_{A_0}(\mathfrak{B}(\rho))=\Tr_{A_1}(\mathfrak{B}(\rho))=\rho\, .
\]

\section{Some preliminary results}

In this section, we prove some preliminary results which are necessary for our proof of the main theorem. We will denote the set of all $n$-chains $(\E_1,\ldots,\E_n)$ from $A_0$ through $A_n$ by $\bold{CPTP}(A_0,\ldots,A_n)$

\begin{definition}
Given an $n$-chain $(\E_1,\ldots,\E_n)\in \bold{CPTP}(A_0,\ldots,A_n)$ and a spatiotemporal product $\star$, we let 
\[
\Phi^{(\E_1,\ldots\E_n)}:A_0\longrightarrow A_1\cdots A_n
\] 
be the mapping given by $\Phi^{(\E_1,\ldots\E_n)}=\mathscr{J}^{-1}[(\E_1,\ldots,\E_n)\star \mathds{1}]$, where $\mathscr{J}:\textbf{Lin}(A_0\, ,A_1\cdots A_n)\to A_0\cdots A_n$ is the \jami isomorphism. The map $\Phi^{(\E_1,\ldots\E_n)}$ will then be referred to as the \emph{\jami map} associated with the $n$-chain $(\E_1,\ldots,\E_n)\in \bold{CPTP}(A_0,\ldots,A_n)$ and $\star$.
\end{definition}

\begin{remark}
We note that it follows from the definition of the \jami map $\Phi^{(\E_1,\ldots\E_n)}$ that
\begin{equation} \label{SWXPJSX67}
\Big(\id_{A_0}\otimes \Phi^{(\E_1,\ldots\E_n)}\Big)(\swap)=(\E_1,\ldots,\E_n)\star \mathds{1}\, .
\end{equation}
\end{remark}

Given a spatiotemporal product $\star$ and a quantum system $A_0$, we let $\mathfrak{B}^{\star}:A_0\to A_0A_1$ be the unique broadcasting map corresponding to the linear extension of the mapping defined on states by $\rho\mapsto \id_{A_0}\star \rho$, where $A_1$ is a copy of $A_0$.

\begin{lemma}
Let $\star$ be a spatiotemporal product which is conditionable, let $A_0$ be a quantum system, let $\Theta_{\rho}:A_0\to A_0$ be the associated state-rendering map for all $\rho\in \mathfrak{S}(A_0)$, and let $\mathfrak{B}^{\star}:A_0\to A_0A_1$ be the associated broadcasting map. Then for all $a\in A_0$, there exists a map $\Theta_{a}:A_0\to A_0$ such that
\begin{equation}\label{SFX359}
\mathfrak{B}^{\star}(a)=(\Theta_{a}\otimes \id_{A_1})(\id_{A_0}\star \mathds{1})\, .
\end{equation}
Moreover, if $a=\sum_i\mu_i\rho^i$ is a state decomposition of $a$, then $\Theta_a=\sum_i\mu_i \Theta_{\rho^i}$.
\end{lemma}

\begin{proof}
Let $a\in A_0$ and let $a=\sum_i \mu_i \rho^i$ be a state decomposition of $a$. Then
\begin{align*}
\mathfrak{B}^{\star}(a)&=\mathfrak{B}^{\star}\Big(\sum_i \mu_i \rho^i\Big)=\sum_i\mu_i\mathfrak{B}^{\star}(\rho^i)=\sum_i \mu_i (\id_{A_0}\star \rho^i)
=\sum_i \mu_i \Big((\Theta_{\rho^i}\otimes \id_{A_1})(\id_{A_0}\star \mathds{1})\Big) \\
&=\left(\Big(\sum_i \mu_i\Theta_{\rho^i}\Big)\otimes \id_{A_1}\right)(\id_{A_0}\star \mathds{1})\, ,
\end{align*}
as desired.
\end{proof}

\begin{lemma} \label{LSMXS71}
Let $\star$ be a spatiotemporal product which is conditionable and state-linear, and let $A_0$ be a finite-dimensional quantum system. Then the following statements hold. 
\begin{enumerate}
\item  \label{LSMXS712}
The \jami map $\Phi^{(\id_{A_0})}$ associated with a spatiotemporal product $\star$ is invertible.
\item \label{LSMXS711}
The state rendering map $\Theta_{\rho}$ is linear in $\rho$, i.e., for all $\rho,\sigma \in \mathfrak{S}(A)$ and for all $a,b\in \mathbb{C}$,
\begin{equation} \label{THETXLXS81}
\Theta_{a\rho+b\sigma}=a\Theta_{\rho}+b\Theta_{\sigma}\, .
\end{equation}
\end{enumerate}
\end{lemma}

\begin{proof}
\underline{Item~\ref{LSMXS712}}:
Let $\Phi^{(\id_{A_0})}$ and $\mathfrak{B}^{\star}:A_0\to A_0A_1$ be the \jami and broadcasting maps associated with the conditionable spatiotemporal product $\star$, and let $\Psi:A_0\to A_0$ be the map given by $\Psi(a)=\Theta_a^{*}(\mathds{1})^{\dag}$, where $\Theta_{a}^*$ is the Hilbert-Schmidt adjoint of the map $\Theta_a$ defined by \eqref{SFX359}. Since $\Theta_{a}$ is a linear combination of state-rendering maps $\Theta_{\rho}$, it follows from item~\ref{LSMXS711} that the map $\Psi$ is linear. For all $a\in A$ we then have
\begin{align*}
a&=\Tr_{A_0}[\mathfrak{B}^{\star}(a)] && \text{since $\mathfrak{B}^{\star}$ is a broadcasting map} \\
&=\Tr_{A_0}\Big[(\Theta_a\otimes \id_{A_1})(\id_{A_0}\star \mathds{1})\Big] && \text{by equation \eqref{SFX359}} \\
&=\Tr_{A_0}\Big[(\Theta_a\otimes \id_{A_1})\Big((\id_{A_0}\otimes \Phi^{(\id_{A_0})})(\swap)\Big)\Big] && \text{by definition of $\Phi^{(\id_{A_0})}$} \\
&=\Tr_{A_0}\Big[\sum_{i,j}\Theta_a(E_{ij})\otimes \Phi^{(\id_{A_0})}(E_{ji})\Big] && \text{since $\swap=\sum_{i,j}E_{ij}\otimes E_{ji}$} \\
&=\sum_{i,j}\Tr[\Theta_a(E_{ij})]\Phi^{(\id_{A_0})}(E_{ji}) && \text{by properties of partial trace} \\
&=\sum_{i,j}\Tr[\Theta_a^*(\mathds{1})^{\dag}E_{ij}]\Phi^{(\id_{A_0})}(E_{ji}) && \text{by definition of $\Theta_a^*$} \\
&=\sum_{i,j}\Phi^{(\id_{A_0})}(\Theta_a^{*}(\mathds{1})^{\dag}_{ji}E_{ji}) && \text{since $\Tr[\Theta_a^*(\mathds{1})^{\dag}E_{ij}]=\Theta_a^*(\mathds{1})^{\dag}_{ji}$}  \\
&=\Phi^{(\id_{A_0})}(\Theta_a^*(\mathds{1})^{\dag}) && \text{since $\Phi^{(\id_{A_0})}$ is linear} \\
&=(\Phi^{(\id_{A_0})}\circ \Psi)(a)\, . && \text{by the definition of $\Psi$}
\end{align*}
It then follows that $\Psi$ is a right inverse of $\Phi^{(\id_{A_0})}$, which is necessarily also a left inverse of $\Phi^{(\id_{A_0})}$ since $\Phi^{(\id_{A_0})}$ and $\Psi$ are both linear, thus $\Phi^{(\id_{A_0})}$ is invertible, as desired.

\underline{Item~\ref{LSMXS711}}: By definition, we have
\begin{equation} \label{eqn:ALem2e1}
    \id\star(a\rho+b\sigma)=(\Theta_{a\rho+b\sigma}\otimes\id_{A_1})(\id_{A_0}\star\mds{1}).
\end{equation}
However, by linearity of $\star$ in $\rho$, we also have
\begin{equation} \label{eqn:ALem2e2}
    \id\star(a\rho+b\sigma)=a(\id\star\rho) + b(\id\star\sigma) = a (\Theta_\rho\otimes\id_{A_1})(\id_{A_0}\star\mds{1})+ b (\Theta_\sigma\otimes\id_{A_1})(\id_{A_0}\star\mds{1}).
\end{equation}
By the proof of Item~\ref{LSMXS712}, $\id_{A_0}\star\mds{1}=(\id_{A_0}\otimes\Psi)(\swap)$ with an invertible $\Psi$. Hence, after applying $\id_{A_0}\otimes \Psi^{-1}$, the equality between \eqref{eqn:ALem2e1} and \eqref{eqn:ALem2e2} yields
\begin{equation}
    (\Theta_{a\rho+b\sigma}\otimes\id_{A_1})(\swap)=\qty(a(\Theta_\rho\otimes\id_{A_1})+b(\Theta_\sigma\otimes\id_{A_1}))(\swap).
\end{equation}
Equation \eqref{THETXLXS81} then follows from the \jami isomorphism, as desired.

\end{proof}

\section{Broadcasting, conditionability and decomposability}\label{SM1}
There are various notions that axiomatize the idea that an initial state and a process represented by a CPTP map can both be distinguished in QSOTs. In this section, we elucidate the interrelation between these notions, which will be crucial for the proof of our main theorem. 

\begin{tcolorbox}[breakable, enhanced jigsaw, colback=white, colframe=black,boxrule=0.2mm]
    \begin{definition} \label{DXBCS67}
    A 1-chain spatiotemporal product $\star$ is said to be
    \begin{itemize}
        \item \textbf{broadcasting} iff for every state $\rho\in \mathfrak{S}(A)$ and every $\E\in \bold{CPTP}(A,B)$,
    \begin{equation}
        \E\star\rho= (\id_A \otimes \E) (\id\star\rho)\, .
    \end{equation}
    \item \textbf{conditionable} iff for every state $\rho\in \mathfrak{S}(A)$ there exists a linear map $\Theta_\rho:A\to A$ such that for every $\E\in \bold{CPTP}(A,B)$,
    \begin{equation}
        \E\star\rho=(\Theta_\rho \otimes \id_B) (\E\star\mds{1}_A)\, .
    \end{equation}
    \item \textbf{decomposable} iff for every state $\rho\in \mathfrak{S}(A)$ and $\E\in \bold{CPTP}(A,B)$ there exist linear maps $\Theta_\rho:A\to A$ and $\Psi^\E:A\to B$ such that
    \begin{equation} \label{DCMBXY87}
        \E\star\rho = (\Theta_\rho\otimes \Psi^\E)(\swap)\, ,
    \end{equation}
    where $\swap$ is the swap operator on two copies of $A$.
    \end{itemize}
    \end{definition}
\end{tcolorbox}

The maps $\Theta_{\rho}$ and $\Psi^{\E}$ appearing in Definition~\ref{DXBCS67} will be referred to as \emph{state-rendering} and \emph{process-rendering} maps, respectively. We note that the state rendering maps $\Theta_{\rho}$ associated with conditionability and decomposability could a priori be different. Given a family of state rendering maps $\Theta_{\rho}$ associated with a spatiotemporal product $\star$, we let $\Theta_{\mathds{1}}$ denote $d_A\Theta_{\pi}$, where $\pi=\mathds{1}/d_A$ is the maximally mixed state on $A$.

\begin{lemma}\label{DDAGXS57}
Let $A$ be a finite-dimensional quantum system, and let $\Theta:A\to A$ be a linear map. Then there exists a unique linear map $\Theta^{\ddag}:A\to A$ such that 
\begin{equation} \label{DDAGXS577}
(\Theta_A\otimes \id_{A})(\swap_{AA'})=(\id_A\otimes \Theta_{A'}^{\ddag})(\swap_{AA'})\, ,
\end{equation}
where $\Theta_{A'}^\ddag$ is the linear map $\Theta^\ddag$ applied to system $A'$ which is a copy of $A$.
\end{lemma}

\begin{proof}
Let $\Theta:A\to A$ be a linear map, let $\mathscr{J}:\bold{Lin}(A,A)\to AA'$ be the \jami isomorphism, let $\gamma:AA'\to AA'$ be the swap isomorphism, let $\mathscr{J}'=\gamma\circ \mathscr{J}$, and let $\Theta^{\ddag}=\mathscr{J}^{-1}[\mathscr{J}'[\Theta]]$. Then
\[
(\Theta\otimes \id_{A})(\swap_{AA'})=\mathscr{J}'[\Theta]=\mathscr{J}[\Theta^{\ddag}]=(\id_A\otimes \Theta^{\ddag})(\swap_{AA'})\, ,
\]
thus \eqref{DDAGXS577} holds. Moreover, $\Theta^{\ddag}$ is the unique map satisfying \eqref{DDAGXS577}, since for any map $\widetilde{\Theta}^{\ddag}$ satisfying \eqref{DDAGXS577} it necessarily follows that $\mathscr{J}[\widetilde{\Theta}^{\ddag}]=\mathscr{J}'[\Theta]=\mathscr{J}[\Theta^{\ddag}]$, which implies $\widetilde{\Theta}^{\ddag}=\Theta^{\ddag}$.
\end{proof}

\begin{lemma} \label{LXMA19}
Conditionability $\implies$ Decomposability $\implies$ Broadcasting.
\end{lemma}

\begin{proof}
Suppose $\star$ is conditionable. Then for every $\E\in \bold{CPTP}(A,B)$ and state $\rho\in \mathfrak{S}(A)$ we have
\begin{align*}
\E\star \rho&=(\Theta_{\rho}\otimes \id_B)(\E\star \mathds{1}) && \text{since $\star$ is conditionable} \\
&=(\Theta_{\rho}\otimes \id_B)\Big((\id_A\otimes \Phi^{(\E)})(\swap)\Big) && \text{by definition of $\Phi^{(\E)}$} \\
&=(\Theta_{\rho}\otimes \Phi^{(\E)})(\swap)\, , && \text{by properties of $\otimes$}
\end{align*}
thus $\star$ is decomposable with process rendering map equal to the \jami map $\Phi^{(\E)}$.

Now suppose $\star$ is decomposable, and let $\Theta_{\rho}$ and $\Psi^{\E}$ be the associated state-rendering and process-rendering maps as defined by \eqref{DCMBXY87} for every state $\rho$ and for every channel $\E$. We then have
\begin{align*}
\E(\rho)&=\Tr_A [\E\star\rho] && \text{by definition of $\star$} \\
&=\Tr_{A}\Big[\Big(\Theta_{\rho}\otimes \Psi^{\E}\Big)(\swap)\Big] && \text{since $\star$ is decomposable} \\
&=\Tr_{A}\Big[\Big(\id_A\otimes \Psi^{\E}\Big)\Big(\id_A\otimes \Theta_{\rho}^{\ddag}\Big)(\swap)\Big] && \text{by Lemma~\ref{DDAGXS57}} \\
&=\Tr_{A}\Big[\sum_{i,j}E_{ij}\otimes \Psi^{\E}(\Theta_{\rho}^{\ddag}(E_{ji}))\Big] && \text{since $\swap=\sum_{i,j}E_{ij}\otimes E_{ji}$} \\
&=\sum_{i,j}\Tr[E_{ij}]\Psi^{\E}(\Theta_{\rho}^{\ddag}(E_{ji})) && \text{by properties of partial trace} \\
&=\sum_{i,j}\delta_{ij}\Psi^{\E}(\Theta_{\rho}^{\ddag}(E_{ji})) && \text{since $\Tr[E_{ij}]=\delta_{ij}$} \\
&=\Psi^\E(\Theta_\rho^{\ddag}(\mds{1}))\, , && \text{since $\sum_i E_{ii}=\mathds{1}$}
\end{align*}
from which it follows that $\E=\Psi^\E\circ \Phi$, where $\Phi:A\to A$ is the map given by $\Phi(\rho)= \Theta_\rho^{\ddag}(\mds{1})$. We note that the map $\Phi$ is linear by item~\ref{LSMXS711} of Lemma~\ref{LSMXS71}. By considering the case $\E=\id$ one may conclude that $\Phi$ is invertible and that $\Psi^\id=\Phi^{-1}$, hence $\Psi^\E=\E\circ \Psi^\id$. We then have
\begin{align*}
\E\star \rho&=(\Theta_\rho\otimes\Psi^\E)(\swap) && \text{since $\star$ is decomposable} \\
&=(\id_A\otimes \E) \left((\Theta_\rho\otimes\Psi^\id)(\swap)\right) && \text{since $\Psi^\E=\E\circ \Psi^\id$} \\
&=(\id_A\otimes \E)(\id_A\star\rho)\, , && \text{since $\star$ is decomposable}
\end{align*}
thus $\star$ is broadcasting, as desired.
\end{proof}

\begin{remark}
It turns out that decomposability with $\Theta_{\mds{1}}$ being the identity map is in fact equivalent to conditionability:

\begin{proof}
\noindent$(\implies)$ Suppose $\star$ is decomposable, let $\Theta_{\rho}$ and $\Psi^{\E}$ be the associated state-rendering and process-rendering maps as defined by \eqref{DCMBXY87} for every state $\rho$ and for every channel $\E\in \bold{CPTP}(A,B)$, and suppose $\Theta_{\mds{1}}=\id_A$. We then have
\begin{align*}
\E\star\rho&=(\Theta_\rho\otimes\Psi^\E)(\swap) && \text{since $\star$ is decomposable} \\
&=(\Theta_\rho\otimes\id_B)\Big((\id_A\otimes\Psi^\E)(\swap)\Big) && \text{by properties of $\otimes$} \\
&=(\Theta_\rho\otimes\id_B)(\Theta_{\mds{1}}\otimes\Psi^\E)(\swap) && \text{since $\Theta_{\mathds{1}}=\id_A$} \\
&=(\Theta_\rho\otimes\id_B)(\E\star\mds{1})\, , && \text{since $\star$ is decomposable}&
\end{align*}
thus $\star$ is conditionable.

\noindent$(\impliedby)$ We have seen from the proof of Lemma~\ref{LXMA19} that conditionability implies decomposability, and that in such a case the state-rendering map $\Theta_{\rho}$ defining conditionability is in fact the state-rendering map defining decomposability. As such, we just have to show that the state-rendering map associated with a conditionable spatiotemporal product is such that $\Theta_{\mathds{1}}=\id_{A}$.  

So suppose $\star$ is conditionable, and let $\Theta_{\mathds{1}}$ be the associated state-rendering map. We then have
\begin{align*}
(\id_A\otimes\Phi^{(\id_A)})(\swap)&=\id_A\star\mds{1} && \text{by equation \eqref{SWXPJSX67}} \\
&=(\Theta_{\mds{1}}\otimes\id_{A})(\id\star\mds{1}) && \text{since $\star$ is conditionable} \\
&=(\Theta_{\mds{1}}\otimes\id_{A})\Big((\id_A\otimes\Phi^{(\id_A)})(\swap)\Big) && \text{by equation \eqref{SWXPJSX67}} \\
&=(\Theta_{\mds{1}}\otimes\Phi^{(\id_A)})(\swap)\, . && \text{by properties of $\otimes$}
\end{align*}
By item~\ref{LSMXS712} Lemma~\ref{LSMXS71} the map $\Phi^{(\id_A)}$ is invertible, thus $(\Theta_{\mathds{1}}\otimes \id_A)(\swap)=\swap$. Now define the map $\mathscr{J}'=\gamma \circ \mathscr{J}$, where $\gamma:AA'\to AA'$ is the swap isomorphism and $\mathscr{J}:\bold{Lin}(A,A)\to AA'$ is the \jami isomorphism. We then have
\[
\mathscr{J}'[\Theta_{\mathds{1}}]=(\Theta_{\mathds{1}}\otimes \id_A)(\swap)=\swap=\mathscr{J}[\id_A]\implies \Theta_{\mds{1}}=\id_A\, ,
\]
where the final implication follows from the fact that $\mathscr{J}'$ is an isomorphism.
\end{proof}
\end{remark}

We now define a mapping associated with a spatiotemporal product which will is fundamental to the notion of iterativity.

\begin{definition}
Given a spatiotemporal product $\star$ and a channel $\E\in \bold{CPTP}(A,B)$, the \emph{bloom} of $\E$ with respect to $\star$ is the linear map $\E^{\star}:A\to AB$ corresponding to the unique linear extension of the map on states given by $\rho\mapsto \E\star \rho$.
\end{definition}

\begin{lemma} \label{PCXDXS537}
Let $\star$ be a spatiotemporal product which is broadcasting, and for every $n$-chain $(\E_1,\ldots,\E_n)\in \bold{CPTP}(A_0,\ldots,A_n)$, let $\Psi^{(\E_1,\ldots,\E_n)}$ be the map given by
\begin{equation}\label{PSISX37}
\Psi^{(\E_1,\ldots,\E_n)}=
\begin{cases}
\E_1 \hspace{7.70cm} \text{if $n=1$}\, , \\
\E_2^{\star}\circ \E_1 \hspace{6.99cm} \text{if $n=2$}\, , \\
(\id_{A_1\cdots A_{n-2}}\otimes \E_n^{\star})\circ \cdots \circ (\id_{A_1}\otimes \E_3^{\star})\circ \E_2^{\star}\circ \E_1 \hspace{1.17cm} \text{if $n>2$}\, .
\end{cases}
\end{equation}
Then for every $n$-chain $(\E_1,\ldots,\E_n)\in \bold{CPTP}(A_0,\ldots,A_n)$ and for every state $\rho\in \mathscr{S}(A_0)$,
\begin{equation}\label{NBXSTX87}
\Big(\id_{A_0}\otimes \Psi^{(\E_1,\ldots,\E_n)}\Big)\Big(\id_{A_0}\star \rho\Big)=\Big((\id_{A_0\cdots A_{n-2}}\otimes \E_n^{\star})\circ \cdots \circ (\id_{A_0}\otimes \E_2^{\star})\circ \E_1^{\star}\Big)(\rho)\
\end{equation}
\end{lemma}

\begin{proof}
For $n=1$ equation \eqref{NBXSTX87} is just the broadcasting condition, which holds by assumption. For $n=2$ we have
\begin{align*}
\Big(\id_{A_0}\otimes \Psi^{(\E_1,\E_2)}\Big)\Big(\id_{A_0}\star \rho\Big)&=\Big(\id_{A_0}\otimes (\E_2^{\star}\circ \E_1)\Big)\Big(\id_{A_0}\star \rho\Big) && \text{by definition of $\Psi^{(\E_1,\E_2)}$} \\
&=\Big(\Big(\id_{A_0}\otimes \E_2^{\star}\Big)\circ \Big(\id_{A_0}\otimes \E_1\Big)\Big)\Big(\id_{A_0}\star \rho\Big) && \text{by properties of $\otimes$} \\
&=\Big(\id_{A_0}\otimes \E_2^{\star}\Big)\Big(\E_1\star \rho\Big) && \text{since $\star$ is broadcasting} \\
&=\Big((\id_{A_0}\otimes \E_2^{\star})\circ \E_1^{\star}\Big)(\rho)\, , && \text{by definition of $\E_1^{\star}$} \\
\end{align*}
thus \eqref{NBXSTX87} holds for $n=2$. The proof for $n>2$ is similar.
\end{proof}

\section{Proof of Main Theorem}

\begin{theorem}
Let $\star$ be a spatiotemporal product which is state-linear and conditionable. Then for every $n$-chain $(\E_1,\ldots,\E_n)\in \bold{CPTP}(A_0,\ldots,A_n)$ with $n>0$, and for every state $\rho\in \mathfrak{S}(A_0)$,
\begin{equation}
(\E_1,\ldots,\E_n)\star \rho=\Big((\id_{A_0\cdots A_{n-2}}\otimes \E_n^{\star})\circ \cdots \circ (\id_{A_0}\otimes \E_2^{\star})\circ \E_1^{\star}\Big)(\rho)\, .
\end{equation}
\end{theorem}

\begin{proof}
Let $n>0$. For every $n$-chain $(\E_1,\ldots,\E_n)\in \bold{CPTP}(A_1\ldots,A_n)$, let $\Phi^{(\E_1,\ldots\E_n)}:A_0\to A_1\cdots A_n$ be the \jami map associated with the $n$-chain $(\E_1,\ldots,\E_n)$ and the spatiotemporal product $\star$, which we recall is defined as
\[
\Phi^{(\E_1,\ldots,\E_n)}=\mathscr{J}^{-1}\Big[(\E_1,\ldots,\E_n)\star \mathds{1}_{A_0}\Big]\, .
\]
Now since $\star$ is conditionable and state-linear, the \jami map $\Phi^{(\id_{A_0})}$ is invertible by item~\ref{LSMXS712} of Lemma~\ref{LSMXS71}, thus we can define a map $\chi^{(\E_1,\ldots,\E_n)}:A_0\to A_1\cdots A_n$ as the composition $\chi^{(\E_1,\ldots,\E_n)}=\Phi^{(\E_1,\ldots,\E_n)}\circ {\Phi^{(\id_{A_0})}}^{-1}$. We then have
\begin{align*}
(\E_1,\ldots,\E_n)\star \mathds{1}_{A_0}&=\Big(\id_{A_0}\otimes \Phi^{(\E_1,\ldots,\E_n)}\Big)(\swap) && \text{by definition of $\Phi^{(\E_1,\ldots,\E_n)}$} \\
&=\Big(\id_{A_0}\otimes \Big(\Phi^{(\E_1,\ldots,\E_n)}\circ {\Phi^{(\id_{A_0})}}^{-1}\circ \Phi^{(\id_{A_0})}\Big)\Big)(\swap) && \text{since $\Phi^{(\id_{A_0})}$ is invertible} \\
&=\Big(\id_{A_0}\otimes \chi^{(\E_1,\ldots,\E_n)}\Big)\left( \Big(\id_{A_0}\otimes \Phi^{(\id_{A_0})}\Big)(\swap)\right) && \text{by definition of $\chi^{(\E_1,\ldots,\E_n)}$} \\
&=\Big(\id_{A_0}\otimes \chi^{(\E_1,\ldots,\E_n)}\Big)\Big(\id_{A_0}\star \mathds{1}\Big)\, , && \text{by definition of $\Phi^{(\id_{A_0})}$} \\
\end{align*}
thus
\begin{equation} \label{CHISTRXS87}
(\E_1,\ldots,\E_n)\star \mathds{1}_{A_0}=\Big(\id_{A_0}\otimes \chi^{(\E_1,\ldots,\E_n)}\Big)\Big(\id_{A_0}\star \mathds{1}\Big)\, .
\end{equation}
We then have
\begin{align*}
(\E_1,\ldots,\E_n)\star \rho&=(\Theta_{\rho}\otimes \id_{A_1\cdots A_n})((\E_1,\ldots,\E_n)\star \mds{1}_{A_0}) && \text{since $\star$ is conditionable} \\
&=\Big(\Theta_{\rho}\otimes \id_{A_1\cdots A_n}\Big)\Big(\Big(\id_{A_0}\otimes \chi^{(\E_1,\ldots,\E_n)}\Big)\Big(\id_{A_0}\star \mathds{1}\Big)\Big) && \text{by equation \eqref{CHISTRXS87}} \\
&=\Big(\id_{A_0}\otimes \chi^{(\E_1,\ldots,\E_n)}\Big)\Big(\Big(\Theta_{\rho}\otimes \id_{A_0}\Big)\Big(\id_{A_0}\star \mathds{1}\Big)\Big) && \text{by properties of $\otimes$} \\
&=\Big(\id_{A_0}\otimes \chi^{(\E_1,\ldots,\E_n)}\Big)\Big(\id_{A_0}\star \rho\Big)\, , && \text{since $\star$ is conditionable} \\
\end{align*}
thus for all $n>0$, for every $n$-chain $(\E_1,\ldots,\E_n)\in \bold{CPTP}(A_0,\ldots,A_n)$ and for every state $\rho\in \mathfrak{S}(A_0)$ we have
\begin{equation}\label{CHXRSX347}
(\E_1,\ldots,\E_n)\star \rho=\Big(\id_{A_0}\otimes \chi^{(\E_1,\ldots,\E_n)}\Big)\Big(\id_{A_0}\star \rho\Big)\, .
\end{equation}
The theorem then follows from Lemma~\ref{PCXDXS537} once we show $\chi^{(\E_1,\ldots,\E_n)}=\Psi^{(\E_1,\ldots \E_n)}$ for all $n$-chains $(\E_1,\ldots,\E_n)\in \bold{CPTP}(A_0,\ldots A_n)$ with $n>0$, where $\Psi^{(\E_1,\ldots \E_n)}$ is given by \eqref{PSISX37}. For this, we use induction on $n>0$. For $n=1$ and for all $\rho\in \mathscr{S}(A_0)$, we have
\begin{align*}
\Psi^{(\E_1)}(\rho)&=\E_1(\rho) && \text{by definition of $\Psi$} \\
&=\Tr_{A_0}\left[\,\E_1 \star \rho\,\right] && \text{by definition of $\star$} \\
&=\Tr_{A_0}\Big[(\id_{A_0}\otimes \chi^{(\E_1)})(\id_{A_0}\star \rho) \Big] && \text{by equation \eqref{CHXRSX347}} \\
&=\chi^{(\E_1)}\Big(\Tr_{A_0}[\,\id_{A_0}\star \rho\,]\Big) && \text{since $\Tr_{A_0}=\Tr\otimes \,\id_{A_1}$} \\
&=\chi^{(\E_1)}(\rho)\, , && \text{by definition of $\star$},
\end{align*}
thus by linearity it follows that $\chi^{(\E_1,\ldots,\E_n)}=\Psi^{(\E_1,\ldots,\E_n)}$ for $n=1$. Now assume $\chi^{(\E_1,\ldots,\E_n)}=\Psi^{(\E_1,\ldots,\E_n)}$ for $n=m-1$ with $m-1>0$, i.e., we assume that for every $(m-1)$-chain $(\F_1,\ldots,\F_{m-1})\in \bold{CPTP}(B_0,\ldots,B_{m-1})$,
\begin{equation} \label{INDXTASXTR737}
\chi^{(\F_1,\ldots,\F_{m-1})}=\Psi^{(\F_1,\ldots,\F_{m-1})}\, .
\end{equation}
In the final step of the proof, we will use the fact that a spatiotemporal product which is conditionable is necessarily broadcasting, which follows from Lemma~\ref{LXMA19}. For all $\rho\in \mathscr{S}(A_0)$ we then have
\begin{align*}
\Psi^{(\E_1,\ldots,\E_m)}(\rho)&=\Big(\Big(\id_{A_1\cdots A_{m-2}}\otimes \E_m^{\star}\Big)\circ \cdots \circ \Big(\id_{A_1}\otimes \E_3^{\star}\Big)\circ \E_2^{\star}\circ \E_1\Big)(\rho) && \text{by definition of $\Psi$} \\
&=\Big(\Big(\id_{A_1\cdots A_{m-2}}\otimes \E_m^{\star}\Big)\circ \cdots \circ \Big(\id_{A_1}\otimes \E_3^{\star}\Big)\Big)\Big(\E_2^{\star}(\E_1(\rho))\Big) && \text{by associativity} \\
&=\Big(\Big(\id_{A_1\cdots A_{m-2}}\otimes \E_m^{\star}\Big)\circ \cdots \circ \Big(\id_{A_1}\otimes \E_3^{\star}\Big)\Big)\Big(\E_2 \star \E_1(\rho)\Big) && \text{by definition of $\E_2^{\star}$} \\
&=\Big(\Big(\id_{A_1\cdots A_{m-2}}\otimes \E_m^{\star}\Big)\circ \cdots \circ \Big(\id_{A_1}\otimes \E_3^{\star}\Big)\Big)\Big((\id_{A_1}\otimes \E_2)(\id_{A_1}\star \E_1(\rho))\Big) && \text{since $\star$ is broadcasting} \\
&=\Big(\id_{A_1}\otimes \Psi^{(\E_2,\ldots,\E_{m})}\Big)\Big(\id_{A_1}\star \E_1(\rho)\Big) && \text{by properties of $\otimes$} \\
&=\Big(\id_{A_1}\otimes \chi^{(\E_2,\ldots,\E_{m})}\Big)\Big(\id_{A_1}\star \E_1(\rho)\Big) && \text{by equation \eqref{INDXTASXTR737}} \\
&=(\E_2,\ldots,\E_{m})\star (\E_1(\rho)) &&  \text{by equation \eqref{CHXRSX347}}  \\
&=\Tr_{A_0}\Big[(\E_1,\ldots,\E_m)\star \rho\Big] && \text{by definition of $\star$} \\
&=\Tr_{A_0}\Big[(\id_{A_0}\otimes \chi^{(\E_1,\ldots,\E_m)})(\id_{A_0}\star \rho)\Big] && \text{by equation \eqref{CHXRSX347}} \\
&=\chi^{(\E_1,\ldots,\E_m)}\Big(\Tr_{A_0}[\,\id_{A_0}\star \rho\,]\Big) && \text{since $\Tr_{A_0}=\Tr\otimes \,\id$} \\
&=\chi^{(\E_1,\ldots,\E_m)}(\rho)\, .  && \text{by definition of $\star$} 
\end{align*}
It then follows from linearity that $\chi^{(\E_1,\ldots,\E_m)}=\Psi^{(\E_1,\ldots \E_m)}$, thus concluding the proof.
\end{proof}

\end{document}